\documentclass{LMCS}

\def\dOi{10(3:6)2014}
\lmcsheading%
{\dOi}
{1--17}
{}
{}
{Oct.~31, 2012}
{Aug.~19, 2014}
{}

\ACMCCS{[{\bf Theory of computation}]: Computational complexity and
  cryptography---Interactive proof systems}


\usepackage{hyperref}

\usepackage{amssymb}
\usepackage{amsmath}
\usepackage{amsthm}
\usepackage{textcomp}
\usepackage{array}
\usepackage{multirow}
\usepackage{color}
\usepackage{comment}

\newcommand{\cent}[0]{\mbox{\textcent}}

\theoremstyle{plain}

\theoremstyle{definition}

\begin{document}

\title[Finite state verifiers with constant randomness]{Finite state verifiers with constant randomness\rsuper*}
\author[A.~C.~C.~Say]{A.~C.~Cem Say\rsuper a}
\address{{\lsuper a}Bo\u{g}azi\c{c}i University, Department of Computer Engineering, Bebek 34342 \.{I}stanbul, Turkey}
\email{say@boun.edu.tr}
\thanks{{\lsuper a}Say was partially supported by T\"{U}B\.ITAK with grant 108E142.}

\author[A.~Yakary{\i}lmaz]{Abuzer Yakary{\i}lmaz\rsuper b}
\address{{\lsuper b}University of Latvia, Faculty of Computing, Raina bulv. 19, R\={\i}ga, LV-1586, Latvia
\hspace{100pt}
\linebreak
National Laboratory for Scientific Computing, Petr\'{o}polis, RJ, 25651-075, Brazil}
\email{abuzer@lncc.br}
\thanks{{\lsuper b}Yakary{\i}lmaz was partially supported by CAPES, T\"{U}B\.ITAK with grant 108E142, FP7 FET-Open project QCS, FP7 FET project QALGO, and ERC Advanced Grant MQC}


\titlecomment{{\lsuper *}Some of the material in this paper was first presented in \cite{SY12}.}
    
\begin{abstract}
We give a new characterization of $\mathsf{NL}$ as the class of languages whose members have certificates that can be verified  with small error in polynomial time by finite state machines that use a constant number of random bits, as opposed to its conventional description in terms of  deterministic logarithmic-space verifiers. It turns out that allowing two-way interaction with the prover does not change the class of verifiable languages, and that no polynomially bounded amount of randomness is useful for constant-memory computers when used as language recognizers, or public-coin verifiers. A corollary of our main result is that the class of outcome problems corresponding to $O(\log n)$-space bounded games of incomplete information where the universal player is allowed a constant number of moves equals $\mathsf{NL}$.
\end{abstract}

\maketitle

\section{Introduction} \label{section:Introduction}

It is known that allowing constant-memory computers to use random bits and to commit small amounts of error increases their power, both as language recognizers \cite{Fr81}, and as verifiers of membership proofs \cite{CL89,DS92}. In this paper, we examine the effects of restricting such probabilistic machines (2pfa's) to use only a constant number of random bits, independent of the length of the input. We prove that such constant-randomness 2pfa's are able to verify membership in precisely the languages in $\mathsf{NL}$. This is an interesting addition to  the facts that $\mathsf{NL}$ has deterministic logspace verifiers, and $\mathsf{NP}$ is the class of languages that has  logspace verifiers that use logarithmically many random bits \cite{CL95}. We obtain this result by demonstrating that such verifiers are equivalent to multihead finite automata. Allowing these constant-coin verifiers to use logarithmic space, and to have two-way interaction with the prover, does not augment the class of verifiable languages. No nonregular language has such an interactive proof system if the verifier is restricted to use public coins. We also  show that, when used as recognizers, no amount of polynomially-bounded randomness gives standard 2pfa's any power beyond their deterministic versions.

The rest of this paper is structured as follows: Section \ref{section:Preliminaries} provides the necessary background. Our results on the new characterization of $\mathsf{NL}$ in terms of finite state verifiers, and the public-coin case, are presented in Section \ref{section:main}.    Several variants of the verifier model are examined in Sections \ref{section:rheads} and \ref{section:2pfak}. A characterization of the class of outcome problems corresponding to $O(\log n)$-space bounded games of incomplete information where the universal player is allowed a constant number of moves is given in Section \ref{section:PAFA}. Section \ref{section:Conclusion} is a conclusion.
\section{Preliminaries} \label{section:Preliminaries}
For background on interactive proof systems with bounds on the usage of space and/or randomness, the reader is referred to \cite{Co93A}. 

The main model of verifier that we will use is a probabilistic Turing machine (PTM) with a read-only input tape and a single read/write work tape. The input tape holds the input string between two occurrences of the end-marker symbol $\cent$, and we assume that the machine's transition function never attempts to move the input head beyond the end-markers. The input tape head is on the left end-marker at the start of the process. The verifier exchanges information with a prover by writing and reading one symbol at a time from the communication alphabet $\Gamma$ in a communication cell. Using this information channel, the prover attempts to prove the membership of the input string in the language under consideration. Of course, one should not trust this blindly, and we even allow the possibility that the prover sends an infinite sequence of symbols, a contingency that could cause careless verifiers to run forever. The machine also has access to a source of random bits. The state set of the verifier PTM is $Q = R \cup D \cup \{q_{a},q_{r}\}$, where $R$ is the set of coin-tossing states, $D$ is the set of deterministic states, and $q_{a}$ (accept) and $q_{r}$ (reject) are the halting states.  One of the non-halting states is designated as the start state. A \textit{configuration} of the verifier is defined to be the 4-tuple consisting of its current internal state, input head position, work tape content, and work tape head position. Associated with each state $q \in Q$, there is a communication symbol $\gamma_q \in\Gamma$. The special ``null symbol'' $\epsilon$ is guaranteed to be a member of $\Gamma$. A ``communication step" starts when any state $q \in R \cup D$ with $\gamma_q \neq \epsilon$ is entered, with $\gamma_q$ being written in the communication cell. The prover can be modeled as a prover transition function $\rho$, which determines the symbol $\gamma \in \Gamma$  to be written in response, based on the input string and the entire communication that has taken place so far.\footnote{This ensures that the prover is not able to detect how many moves have been executed by the verifier up to the present communication step.}  Let $\lozenge=\{-1,0,+1\}$ denote the set of possible head movement directions. When the verifier reads the response of the prover, it behaves according to the  verifier transition function $\delta$ as follows: For $q \in R$, $\delta(q,\sigma,\theta,\gamma,b)=(q',\theta',d_{i},d_{w})$ indicates that the machine will switch to state $q'$, write $\theta'$ on the work tape, move the input head in direction $d_{i} \in \lozenge$, and the work tape head in direction $d_{w} \in \lozenge$, if it is originally in state $q$, scanning the symbols $\sigma$,  $\theta$, and $\gamma$ in the input and work tapes, and the communication cell, respectively, and seeing the random bit $b$ as a result of the coin toss. For $q \in D$, $\delta(q,\sigma,\theta,\gamma)=(q',\theta',d_{i},d_{w})$ has a similar meaning, but without the randomness. If $\gamma_q = \epsilon$, the verifier transition function described above is applied directly, without any communication.

There are two different definitions of interactive proofs for language membership. We start with the ``strong" definition.

We say that language $L$ \textit{has a (private-coin) interactive proof system (IPS) with error probability} $\varepsilon$ if there exists a prover $P$ and a verifier $V$ such that
\begin{enumerate}
	\item [1.] for every $x \in L$, the interaction of $P$ and $V$ on input $w$ results in acceptance with probability at least $1-\varepsilon$, and,
	\item [2.] for every $x \notin L$, and for any prover $P^*$, the interaction of $P^*$ and $V$ on input $x$ results in rejection with probability at least $1-\varepsilon$.
\end{enumerate}

\noindent Interactive proof systems where the verifier accepts every member of the language with probability $1$ are said to have \textit{perfect completeness}.

$\mathsf{IP}$ is the class of languages that have  interactive proof systems with polynomial-time verifiers, and with error probability $\varepsilon$ for some $\varepsilon < \frac{1}{2}$. $\mathsf{IP}(\langle$restriction 1$\rangle,\cdots,\langle$restriction $k \rangle)$ will denote the class of languages that have IPSs with $\langle$restriction 1$\rangle,\cdots,\langle$restriction $k \rangle$ \cite{Co93A,CL95}.  We will be examining restrictions on expected runtime, worst-case space (i.e. work tape cells), and random bits. We use the notations $cons$, $log$,  $poly$, and $exp$ to stand for functions that are $O(1)$, $O(\log n)$, $O(n^{c})$, and $2^{O(n^{c})}$ for any constant $c$, respectively. For instance, $\mathsf{IP}=\mathsf{IP}($poly-time). 

We will also be considering the effects of restricting tape head movement on our models. In the general case, both the input and the work tape heads are allowed to move in both directions (except when the input head is on an end-marker) or to stay put, as represented by the set $\lozenge=\{-1,0,+1\}$ in the definition above. A machine where a particular head is not allowed to move left is said to have \textit{one-way} access to the corresponding tape. Heads which are restricted even further so that they are not allowed to stay put, and must move right at  every step, are called \textit{real-time}. These features will be represented by the notations $ \mathsf{1way\mbox{-}input} $ and $ \mathsf{rt\mbox{-}input} $, respectively, in the restriction lists in class names.

Replacing condition 2 in the definition above with the weaker condition 
\begin{enumerate}
	\item [$2'$.]for every $x \notin L$, and for any prover $P^*$, the interaction of $P^*$ and $V$ on input $x$ results in acceptance with probability at most $\varepsilon$
\end{enumerate}
leads to our definition of the $\mathsf{IP}_w\mathsf{(}$restriction-list) classes, the counterparts of $\mathsf{IP(}$restriction-list) with these alternative kinds of verifiers that do not have to halt with high probability for all inputs. Note that, since these ``weak" IPS's are less constrained than the ``strong" ones of the previous definition, the $\mathsf{IP}$ classes are always contained in the corresponding $\mathsf{IP}_w$ classes.  

A \textit{one-way interactive proof system} \cite{Co93B} is an IPS where the prover is restricted so that it maps the set of input strings to the set of sequences from the communication alphabet.\footnote{Note the terminological accident here. We have already used the word ``one-way" to describe tape heads which can not move to the left. Unfortunately, IPSs where the prover does all the talking also happen to be called with this name in the literature. We will be clear about which feature we are referring to throughout the paper.} For input string $w$, the prover writes the $i$th symbol of the corresponding sequence in the communication cell at the $i$th time  the verifier enters a state $q$ with $\gamma_q \neq \epsilon$. This ensures that the communication between the prover and the verifier is one-way. The corresponding language classes are named by prefixing the class names mentioned above with the designation ``oneway-". Note that a one-way IPS can be modeled as a verifier which has one-way access to an additional ``certificate" tape, on which a purported membership proof of the input string has been written, without the need to mention a prover or a communication cell at all, as in the definitions of conventional nondeterministic classes.

The following equalities are trivial: 
\begin{equation}
       \label{equation:detpolyspace}
     \mathsf{NP}=\mathsf{oneway\mbox{-}IP}\mbox{(poly-time, poly-space, 0-random-bits)}
\end{equation}
\begin{equation}
       \label{equation:detlogspace}
      \mathsf{NL}=\mathsf{oneway\mbox{-}IP}\mbox{(poly-time, log-space, 0-random-bits)}
\end{equation}
Note that specifying 0 as the randomness complexity of the verifier is just a way of saying that it is deterministic.

Allowing logarithmic amounts of randomness yields the characterization \cite{CL95}
\begin{equation}
       \label{equation:randlogspace}
       \begin{array}{rcl}
    	   	\mathsf{NP} & = & \mathsf{oneway\mbox{-}IP}\mbox{(poly-time, log-space, log-random-bits)}
    	   	\\
	       	& = & \hspace*{37pt} \mathsf{IP}\mbox{(poly-time, log-space, log-random-bits)}
       \end{array}       
\end{equation}
with an improvement in the space bound. Relaxing the randomness bound of the one-way IPS further does not help on its own, since \cite{Co93B} 
\begin{equation}
       \label{equation:co93}
       \mathsf{oneway\mbox{-}IP}\mbox{(poly-time, log-space)}=\mathsf{NP},
\end{equation}
but allowing interaction as well famously yields \cite{Co91,Sh92}
\begin{equation}
       \label{equation:pspace}
     \mathsf{IP}\mbox{(poly-time, log-space)}=\mathsf{PSPACE}.
\end{equation}

A \textit{public-coin IPS}, also known as an \textit{Arthur-Merlin game}, is an IPS where the coins of the verifier can be seen by the prover when they are flipped, thereby ensuring that the prover always knows the verifier's configuration during the communication. The public-coin version of  $\mathsf{IP}(\langle$restriction 1$\rangle,\cdots,\langle$restriction $k \rangle)$ will be named $\mathsf{AM}(\langle$restriction 1$\rangle,\cdots,\langle$restriction  $k \rangle)$, and the notation will also be extended to the weak definition in a similar way. It is known \cite{Co89,GS89,Sh92} that
\begin{equation}
       \label{equation:amp}
       \mathsf{AM}\mbox{(exp-time, log-space)}=\mathsf{P},
\end{equation}
and
\begin{equation}
       \label{equation:amip}
     \mathsf{AM}\mbox{(poly-time, poly-space)}=\mathsf{PSPACE}.
\end{equation}

The relationships in Equations \ref{equation:detpolyspace}-\ref{equation:amip} remain true for the weak definition of IPS's, since logarithmically bounded space is sufficient to cut off unacceptably long computational paths. When one considers finite state verifiers, \cite{DS92}
which use only a constant amount of cells on the work tape,\footnote{It is easy to see that such machines can be simulated by machines with longer programs which have no work tape at all, namely, two-way probabilistic finite automata (2pfa's) \cite{Fr81}.} the difference between the weak and strong definitions becomes evident. 

With no limits on the runtime, or the number of random bits to be used, weak IPS's with finite state verifiers exist for a vast class of languages; $\mathsf{oneway\mbox{-}IP}_w\mbox{(cons-space)}$ contains every recursively enumerable language, whereas $\mathsf{IP}\mbox{(cons-space)}$ is contained in $\mathsf{SPACE(2^{2^{O(n)}})}$ \cite{CL89}. It has been proven \cite{DS92} that Arthur-Merlin games with finite state verifiers exist for languages outside  the class of languages recognizable by ``stand-alone'' 2pfa's, and that some languages have linear-time finite state verifiers only if the public-coin restriction is not enforced, in contrast to Equations \ref{equation:pspace} and \ref{equation:amip}.

We will focus on verifiers which use a constant number of random bits for any input.

In the next section, we will demonstrate an interesting relationship between constant-space, constant-randomness  verifiers and multihead finite automata. A $k$-head finite automaton (2nfa($k$)) is simply a nondeterministic finite-state machine with $k$  two-way heads that it can direct on a read-only tape containing the input string, flanked by two end-markers. A configuration of a 2nfa($k$) is a tuple consisting of its current state and head positions. Deterministic multihead finite automata (2dfa($k$)'s) are defined analogously. The classes of languages recognized by these machine families will be denoted as $\mathsf{2NFA}$($k$) and $\mathsf{2DFA}$($k$), respectively.  We will also look at probabilistic    versions of multihead automata (2pfa($k$)'s). Detailed information about these machines can be found in \cite{HKM11,Ma97}. 
 We note the following important facts that will be used in our proofs.

\begin{fact}
	\label{fact:mhfa-nl}  $\bigcup_{k\geq 1}\mathsf{2NFA}(k)=\mathsf{NL}$ \cite{Ha72}.
\end{fact}
\begin{fact}
	\label{fact:2dfak-ll}  $\bigcup_{k\geq 1}\mathsf{2DFA}(k)=\mathsf{L}$ \cite{Ha72}.
\end{fact}
\begin{fact}
	\label{fact:2k} Every 2nfa($k$) (resp., 2dfa($k$)) has an equivalent
2nfa($2k$) (resp., 2dfa($2k$)) that halts  in $O(n^k)$ time  on every computational branch.\footnote{We thank Martin Kutrib, who taught us the proof of this fact.}
\end{fact}

Multihead finite automata where all heads are restricted to one-way movement (\textit{one-way $k$-head  automata}) will be denoted 1nfa($k$)'s. The corresponding language classes are named $\mathsf{1NFA}$($k$). The probabilistic and deterministic versions of these machines will be denoted 1pfa($k$) and 1dfa($k$), respectively.


\section{2pfa verifiers with constant randomness and 2nfa($k$)'s} 
\label{section:main}
We start our examination of the effects of limiting the number of random bits by noting that machines that are not helped by a prover about their input are very weak when restricted to work with constant workspace, and polynomially bounded randomness.

\begin{thm}\label{theorem:polyrand}
For any polynomial $p$, every 2pfa whose expected number of coin tosses on halting computational branches is $O(p(n))$ for input strings of length $n$ recognizes a regular language with bounded error.
\end{thm}
\begin{proof}
This is a straightforward modification of the proof (in \cite{DS90})
of the following fact \cite{DS90,KF91}:
\begin{center}
\begin{minipage}{0.95\textwidth}
       \textsl{For any polynomial $p$, 2pfa's with expected runtime $O(p(n))$
       recognize  only the regular languages with bounded error.}
\end{minipage}
\end{center}
See  Appendix \ref{appendix:one} for the details.
\end{proof}

Our new characterization of $\mathsf{NL}$ is demonstrated by the following  lemmas.

\begin{lem} \label{lemma:mhfatoip}
For any language $L$ in $\mathsf{NL}$, there exists a weak one-way, constant-space, constant-randomness IPS that recognizes $L$ with perfect completeness for any desired error probability $\varepsilon<\frac{1}{2}$.
\end{lem}

\begin{proof}
By Fact \ref{fact:mhfa-nl}, $L$ is recognized by a 2nfa($k$) $M$. We show how to construct an IPS with the required properties. As mentioned above, this is equivalent to demonstrating how every member of $L$ has a membership certificate that can be checked with such a verifier. We start by building a verifier $V$ that simulates one run of $M$, by consulting the certificate for choosing among the nondeterministic branches of $M$. $V$ uses just $r = \lceil\log k \rceil$ random bits to branch to $k$ computation paths (each path has probability at least $2^{-r}$) while scanning the left input end-marker. Each such path will use its head to track the position of the corresponding head of $M$. For every step of the simulation of $M$, the certificate contains a symbol conveying the list of $k$ symbols that would be scanned by $M$'s heads at this step, together with an indication of which nondeterministic choice should now be taken by $M$ to eventually reach the accept state. The $i$th path of $V$ rejects immediately if it sees that the present certificate symbol is inconsistent with what the $i$th head is currently scanning, and updates its state and head position according to $M$'s program and the information given by the certificate otherwise.

If the input string is accepted by $M$, the certificate will lead all paths of $V$ to acceptance, by giving correct information about what the heads are seeing and the nondeterministic choice at every step, yielding a total acceptance probability of $1$. Otherwise, any certificate must ``lie'' about at least one head in order to make some paths  accept, causing the path responsible for that head to reject, so the acceptance probability in that case is at most $1-2^{-r}$. To reduce the unacceptably high error bound for nonmembers, we chain several copies of $V$ to run one after another,\footnote{Note that the certificate guides the paths of $V$ to position their heads back on the left end-marker and to start the next round of coin-flipping simultaneously.} on a correspondingly long certificate, and accept if and only if all copies accept, rejecting otherwise. It is easy to see that a chain of $m$ copies of $V$ involves an error of $ (1-2^{-r})^{m} $, and therefore $m \ge \frac{\log\varepsilon}{\log(1-2^{-r})}$ iterations are sufficient to obtain an error of $\varepsilon$, where the total number of random bits used by the resulting verifier would be $O( k \log k \log\frac{1}{\varepsilon})$. Note that a 2nfa($k$) with state set $Q$ has at most $|Q|(n+2)^{k}$ distinct reachable configurations on any input of length $n$, and therefore $V$ runs in polynomial time for correct proofs of membership.
 \end{proof}

This result enables us to determine the minimum number of ``useful" random bits required by 2pfa verifiers: A single coin toss would create just two computational paths with equal probability. Since a probabilistic machine that always responds correctly can be replaced by its deterministic counterpart, we must have the verifier err for at least one input string to have any hope of outperforming a two-way deterministic finite automaton. But the probability of such an error is at least $\frac{1}{2}$ in a  machine that tosses its coin only once, which would violate our bounded error condition. Additional random bits can be used to reduce the error probability as described in the proof of Lemma \ref{lemma:mhfatoip}, and $\mathsf{2NFA}(2)$, which contains nonregular languages, has verifiers with two random bits.

The reader should also note that the IPSs of Lemma \ref{lemma:mhfatoip} are strictly more powerful than  2pfa's  unaided by a prover, even when the latter are allowed to use an unbounded number of fair coins, since it is known \cite{Ka89,Ma98} that the class of languages recognizable by such stand-alone 2pfa's is properly contained in the class $\mathsf{L}$.

 The reason why the construction in Lemma \ref{lemma:mhfatoip} does not yield an IPS according to the strong definition is that an evil prover can supply an infinitely long fake certificate that makes some paths of the verifier enter infinite loops by lying\footnote{We can assume that the simulated multihead automaton has the  desirable property mentioned in Fact \ref{fact:2k}. Any prover that causes a long runtime must therefore be lying.} about a head that those paths cannot see, at the cost of being rejected by the path responsible for that head. If we forgo the guarantee of halting with probability 1 for members of the language, (thereby losing perfect completeness,) and the capability of reducing the error bound to any desired nonzero value, settling for an $\varepsilon$ that is near (but of course strictly less than) $\frac{1}{2}$, we can create a strong one-way IPS for any language in $\mathsf{NL}$, as the next lemma shows.
 
 \begin{lem} \label{lemma:strondef} $\mathsf{NL}  \subseteq  \mathsf{oneway\mbox{-}IP}\mathrm{(cons\mbox{-}space, cons\mbox{-}random\mbox{-}bits)}$.
\end{lem}

\begin{proof}
Let $L$ be any language in $\mathsf{NL}$. We first use the construction in the beginning of the proof of   Lemma \ref{lemma:mhfatoip} to build a verifier $V$ that uses $r= \lceil\log k \rceil$ random bits to simulate one run of the  2nfa($k$) associated with $L$, accepts correct certificates  for members of $L$ with probability 1, and rejects any incorrect certificate with probability at least $2^{-r}$. We then augment $V$ to obtain a new verifier $V'$, which uses $r+1$ more random bits, as follows: $V'$ rejects directly with probability $\frac{2^{r}-1}{2^{r+1}}$. With the remaining probability, $V'$ transfers control to $V$.

$V'$ accepts correct membership certificates with probability $\frac{2^{r}+1}{2^{r+1}}$, i.e. with an error of $\frac{2^{r}-1}{2^{r+1}}$. Any incorrect certificate is rejected with probability at least $\frac{2^{2r}+1}{2^{2r+1}}$, yielding a strong IPS with error bound $\frac{2^{2r}-1}{2^{2r+1}}$. Note that an honest prover can always supply a certificate that causes $V'$ to halt with probability 1 for members of $L$ within polynomial time; and $V'$ can be tricked to running forever by evil provers only with probability at most $\frac{2^{2r}-1}{2^{2r+1}}$.
\end{proof}

We will now show that two-way interaction with the prover does not augment the power of constant-randomness verifiers, even if they are allowed to use logarithmic space, and no requirement of halting with probability 1, let alone a time bound, is imposed on computations for inputs in the language.
\begin{lem} 
	\label{lemma:iptomhfa} $\mathsf{IP}_w\mathrm{(log\mbox{-}space, cons\mbox{-}random\mbox{-}bits)} \subseteq \mathsf{NL}$. 
\end{lem}
\begin{proof}
We start by showing that any language in $\mathsf{IP}_w\mathrm{(log\mbox{-}space, cons\mbox{-}random\mbox{-}bits)}$ has an IPS with a worst-case polynomial  bound on the runtime of the verifier.

Suppose that a language $L$ has a weak IPS with error $\varepsilon$ consisting of prover $P$ and logspace verifier $ V $, which  always uses at most $r$ random bits. Assume without loss of generality that $V$ tosses all of its coins at the start, and then transfers control to the appropriate member of S=$\{M_1,M_2,\ldots,M_{2^{r}}\}$, where each $M_i$ is a deterministic logspace verifier corresponding to the $i$th possible assignment to the $r$-bit random string. The prover $P$ can be viewed as communicating with these deterministic verifiers, and eventually convincing more than half of them to accept the input strings in $L$. Note that the number of distinct reachable configurations  of any of the $M_i$ is bounded by a polynomial, say, $c(n)$, in the input length $n$, and these machines can therefore run for at most $c(n)$ steps between any two consecutive communication steps.

Some members of S can have the same communication transcript, that is, they can send precisely the same sequence of symbols to $P$, and therefore receive the same sequence of responses. Since this is a private-coin system, $P$ does not know which particular $M_i$ it is talking to in such cases. From the point of view of $P$, the state of $V$ at any communication step is a probabilistic mixture (an ``ensemble") of the configurations of the deterministic verifiers consistent with the interchange so far. Since such an ensemble can contain at most $2^r$ elements, the total number of possible ensembles is itself bounded by a polynomial, say, $p(n)$, in $n$. We therefore conclude that $P$ does not need to communicate more than $p(n)$ symbols to convince $V$ for any input string in $L$, since a longer communication would necessarily repeat an ensemble and can be shortened without changing the result. It follows that all accepting branches of $V$ have polynomially bounded runtime for all members of $L$ when communicating with such a $P$.  One can, if one wishes, then build a new logspace verifier that simulates $V$, rejecting when the execution of any $M_i$ has exceeded this time bound, to obtain a new IPS handling the language $L$ with the same error and randomness cost.

Let us now proceed to show that $L \in \mathsf{NL}$, by building a   one-way IPS with a deterministic logspace verifier $ M $ for $L$ (recalling Equation \ref{equation:detlogspace}). We will use the probabilistic verifier $V$ described above in our construction. Let us say that $V$ and $P$ use the communication alphabet $\Gamma$, and that $V$ has state set $Q=C\cup N$, where $C=\{q \in Q \mid \gamma_q\neq \epsilon \}$ is the set of states which communicate with the prover, and $N$ is the set of ``noncommunicating" states. Recall that $V$ starts by randomly picking a member of the set S of $2^{r}$  deterministic verifiers, the $M_i$. 

The purported membership certificate that our new verifier $M$ will check consists of $2^r$ tracks, each with alphabet $T = \Gamma\cup \{\bigcirc,\infty\}$. The $i$th track is supposed to contain a transcript of $M_i$'s communications with the prover about  the  input $x$. The $i$th track square of the $j$th certificate symbol contains 
\begin{itemize}
\item $\gamma$, if $M_i$ receives the prover response $\gamma$ in its $j$th communication step,
\item $\bigcirc$, if $M_i$ performs a halting computation with fewer than $j$ communications, and,
\item $\infty$ otherwise, that is, if $M_i$ enters a nonhalting path of noncommunicating states after performing fewer than $j$ communications.
\end{itemize}\smallskip

\noindent To process the $j$th certificate symbol, $M$ simulates all the $M_i$'s that are indicated to be on a halting path on the input $x$ until they reach their $j$th communication step,  terminate, or are detected to have entered an infinite loop by running more than $c(n)$ steps. $M$ rejects if  it detects a mismatch between the track content and the actual computation of $M_i$.

Recall that some members of S can have the same communication transcript, and are therefore indistinguishable by the prover, for the input at hand. Partition S into blocks, each of which correspond to a different  communication transcript. 
$M$ discovers this partition as it goes through the certificate. At the start, it considers all the $M_i$'s as in the same block in the initial partition. Whenever it scans a new certificate symbol,  $M$ refines the partition to separate the $M_i$'s that send different symbols, or perform no communication, and rejects if the certificate is claiming that different prover messages are being received by two verifiers in the same block of the new partition. If any track contains a communication symbol after the appearance of a $\bigcirc$ or an $\infty$, $M$ rejects. If it detects that the certificate is longer than $p(n)$ steps, $M$ rejects. $M$ accepts if the certificate survives these tests, and a majority of the $M_i$'s are verified to terminate with acceptance.

Clearly, a majority of the members of S accept as a result of their interaction with $P$ on the input $x$ if and only if $x\in L$. If the input is not in $L$, there is no prover that can fool  $V$ for more than half of its possible coin strings to cause acceptance together, and no certificate can make $ M $ accept  this input.
 \end{proof}

We have proven that
\begin{thm} \label{theorem:main} 
\[
	\mathsf{oneway\mbox{-}IP}\mathrm{(cons\mbox{-}space, cons\mbox{-}random\mbox{-}bits)}= \mathsf{NL}=\mathsf{IP}_w\mathrm{(log\mbox{-}space, cons\mbox{-}random\mbox{-}bits)}.
	\]
\end{thm}\smallskip

\noindent Let $\mathsf{REG}$ denote the set of regular languages. We also have the following to say about the public-coin versions of these verifiers.

\begin{fact}\label{fact:AMwconscons} $\mathsf{AM}_w\mathrm{(cons\mbox{-}space, cons\mbox{-}random\mbox{-}bits)} = \mathsf{REG}$.
\end{fact}
This is a special case of the following theorem.
\begin{thm}\label{theorem:AMwgen} For any resource bounds $r(n)$ and $s(n)$, where $r(n)$ is computable in space $s(n)$, and $r(n)\in O(s(n))$, $\mathsf{AM}_w\mathrm{(}s\mathrm{\mbox{-}space,\mbox{ }}r\mathrm{\mbox{-}random\mbox{-}bits)} = \mathsf{NSPACE}($s$\mathsf{)}$.
\end{thm}
\begin{proof}
One direction is obvious. For the other direction, we adapt the proof of Lemma \ref{lemma:iptomhfa}. Let $V$ be a public-coin verifier utilizing $r(n)$ coins and $s(n)$ space for a language $L$. We build a deterministic verifier $M$. Let the  set S consist of the $2^{r(n)}$ deterministic verifiers (the $M_i$'s in the terminology of Lemma \ref{lemma:iptomhfa}) obtained by hardwiring all possible coin sequences to $V$. Since the prover is now free to send different messages to each of these verifiers, we do not have to worry about checking for consistency among the supplied communication transcripts of those machines. $M$ can therefore simulate them sequentially, rather than in parallel, requiring its certificate to just present the transcripts of the communication between each $M_i$ and the prover one after another. This certificate can be controlled  in $s(n)$ space.
\end{proof}

\section{Restrictions on heads} \label{section:rheads}
%

\subsection{One-way verifiers} \label{subsection:1way}

In this section, we will show that a relationship similar to the one established in Section \ref{section:main} exists between verifiers that are further restricted to perform one-way access to their input string, and the  family of one-way multihead machines, the 1nfa($k$)'s.

\begin{thm}\label{theorem:oneway}
\[ 
	\begin{array}{l}
		\mathsf{oneway\mbox{-}IP}\mathrm{(cons\mbox{-}space, cons\mbox{-}random\mbox{-}bits,1way\mbox{-}input)} =
		\\
		\hspace*{38pt}
		\mathsf{IP}\mathrm{(cons\mbox{-}space, cons\mbox{-}random\mbox{-}bits,1way\mbox{-}input)}
		= \bigcup_{k\geq 1}\mathsf{1NFA}(k).
	\end{array} 
\]
\end{thm}
\begin{proof}
One of the nontrivial inclusions is easy to prove: Replace the 2nfa($k$) mentioned in the proof of Lemma \ref{lemma:strondef} with any 1nfa($k$), and the construction there yields an equivalent one-way IPS with a constant-space, constant-randomness verifier that has a one-way input head.

For the remaining inclusion, suppose that we are given an IPS with a verifier $ V $,  which always uses at most $r$ random bits, 		and a one-way input head, for a language $L$. We start by transforming $ V $
 to a set S of $2^{r}$ 1dfa verifiers, each of which simulates a version of $ V $ with a different assignment to the $r$-bit random string.

We build a 1nfa($2^{r}$) $ M $ to recognize $L$. As in the proof of Lemma \ref{lemma:iptomhfa}, $ M $ guesses a certificate, and simulates $ V $ to see if this certificate describes a correct transcription of a dialogue of $V$ with the prover that ends with the input string being accepted with high probability. $M$ uses a different head for representing the head position of each machine in set S.  For each newly guessed certificate symbol $\gamma$, $M$ goes through all the machines in S. Each such machine $A$ can either spend a finite number of steps without communicating with the prover, or enter an infinite loop with no further communication. The number of distinct configurations of $A$ in this situation equals the number of internal states of $V$, so $M$ can detect if $A$ has entered such a loop easily. In this manner, $M$ simulates $A$ until it determines that $A$ is looping, or has halted, or has communicated. $M$ checks the certificate for consistency with the information available to the prover, counts the number of the elements of S that are observed to accept for legitimate certificates, and halts and accepts if this counter reaches $2^{r-1}+1$.

If the input is in $L$, then a prover convinces a majority of the machines in S to accept. $ M $ would then have an accepting computation path corresponding to that prover. If the input is not in $L$, there is no prover that can fool more than half of $V$'s paths to accept together, and $ M $ therefore has no accepting path for this input.
\end{proof}

The  family of one-way multihead automata is known \cite{HKM11} to recognize a proper subclass of $\mathsf{NL}$ that properly contains the regular languages. For instance, the language of binary palindromes is not a member of this class, but its complement is. We can therefore conclude that restricting the input head of a verifier to one-way movement does reduce its overall computational power under these resource bounds.

\subsection{Real-time verifiers} \label{subsection:rt}

The class of languages recognized by real-time nfa($k$)'s is precisely the class of regular languages, since having multiple real-time heads on the same tape is no different than having a single head. In contrast, we will now show that constant-randomness finite-state verifiers with real-time access to their input can verify membership in some non-context-free languages. Consider  the language $ \mathtt{TWIN} =\{wcw \vert w\in \{a,b\}^{*}\}$ on the alphabet $\{a,b,c\}$.

\begin{thm}\label{theorem:rt} 
 $ \mathtt{TWIN} \in \mathsf{oneway\mbox{-}IP}\mathrm{(cons\mbox{-}space, cons\mbox{-}random\mbox{-}bits,rt\mbox{-}input)}.$ 
\end{thm}
\begin{proof}
We describe the verifier. Use a random bit to split to two branches on the left end-marker. The first branch immediately starts reading the certificate and comparing it with the prefix of the input that is followed by the first $c$, whereas the second branch does not consult the certificate until it sees a $c$ in the input. The first branch rejects if the comparison fails, or if it sees that the number of $c$'s in the input is not 1. The second branch compares the certificate with the input suffix after the first $c$, and rejects if that comparison fails. Both branches use two more random bits during the execution,\footnote{For instance, they may flip coins when scanning the $c$ symbol and the right end-marker.} and reject if both these bits turn out to be zero. They otherwise accept.

Members of $\mathtt{TWIN}$ will be verified to be so with probability $\frac{3}{4}$ when the certificate is the substring appearing on either side of the $c$. No certificate can convince the verifier with probability greater than $\frac{3}{8}$ when the input is not in $\mathtt{TWIN}$.
\end{proof}

{\sloppy
Note that $\mathsf{oneway\mbox{-}IP}\mathrm{(cons\mbox{-}space, cons\mbox{-}random\mbox{-}bits,rt\mbox{-}input)}$ also contains some nonstochastic languages, not recognizable 
by 2pfa's with unbounded error.
For example, by using an argument similar to the one for $ \mathtt{TWIN} $, it is not hard to show that
$ \mathtt{NH} \in   \mathsf{oneway\mbox{-}IP}\mathrm{(cons\mbox{-}space, cons\mbox{-}random\mbox{-}bits,rt\mbox{-}input)}$, where \cite{NH71}
\begin{equation*}
	\mbox{	
	$ \mathtt{NH} = \{a^{x}ba^{y_{1}}ba^{y_{2}}b \cdots a^{y_{t}}b \mid x,t,y_{1}, \cdots, y_{t}
		\in \mathbb{Z}^{+} \mbox{ and } \exists k ~ (1 \le k
                \le t), x=\sum_{i=1}^{k}y_{i} \}. $
	}
\end{equation*}
}

In our results above, although the input head was real-time, the certificate head was not, and the algorithms used its capability to stay put in some steps critically. Restricting the verifiers to having real-time access on both the input and certificate tapes would indeed reduce the class of languages with constant-space, constant-randomness to $\mathsf{REG}$, since the construction of Theorem \ref{theorem:oneway} can be adapted to obtain an equivalent 1nfa($k$) with all heads working real-time for such a verifier.

\section{2pfa's with multiple heads} \label{section:2pfak}

We have seen that constant randomness seems to convey the power of possessing multiple heads to single-head verifiers. We now turn to machines that already have multiple heads. 
In this section, we examine the effect of limiting the number of allowed coin tosses of 2pfa($k$)'s, that is, probabilistic multihead automata. This will turn out to be relevant for understanding the relationship between the class $\mathsf{L}$ and its randomized generalizations.

\subsection{Two-way heads} \label{subsection:2wayheads}

Hartmanis' proof of Facts \ref{fact:mhfa-nl} and \ref{fact:2dfak-ll} \cite{Ha72} is based on a demonstration  that multihead finite automata and logarithmic-space Turing machines can simulate each other with polynomial slowdown. This interchangeability, which extends to the probabilistic versions of these models as well \cite{Ma97}, will be useful for our analyses in this section.

Let us consider the minimum amount of useful randomness for stand-alone 2pfa($k$)'s. Of the several modes of recognition associated with probabilistic machines, (i.e. with one-sided or two-sided, bounded or unbounded error), we take the least restricted one, namely, two-sided unbounded error, where all and only the strings that are members of the language in question are to be accepted with any probability greater than $\frac{1}{2}$. We will use the following variant of Theorem \ref{theorem:AMwgen}.

\begin{thm}\label{theorem:derandomTM}
For any resource bounds $r(n)$ and $s(n)$, where $r(n)$ is computable in space $s(n)$, $s(n)$ is space constructible, and $r(n)\in O(s(n))$, 
the class of languages recognized with PTM's that use at most $r(n)$ random bits and $s(n)$ space is contained in $\mathsf{SPACE(}s\mathsf{)}$.
\end{thm}
\begin{proof}
We start with a probabilistic Turing machine $M$ with the properties mentioned in the statement of the theorem. We build a deterministic Turing machine $D$ as follows. $D$ computes $r(n)$, and starts to simulate all the $2^{r(n)}$ deterministic Turing machines that correspond to different coin sequences of $M$ sequentially on the input. Simulations that are detected to enter infinite loops (by running more than $2^{s(n)}$ steps) are cut off. $D$ counts the simulations that are seen to accept, and accepts if and only if this value exceeds $2^{r(n)-1}$.
It is clear that $D$ uses $O(s(n))$ space, and recognizes the language of $M$.
\end{proof}
Assume that we are given a 2pfa($k$), say, $P$, that uses at most logarithmically many random bits. There exists a logarithmic-space PTM, say, $T$, that uses precisely the same number of coins, and recognizes the same language as $P$ with unbounded error \cite{Ma97}. By Theorem \ref{theorem:derandomTM}, this language is in the class $\mathsf{L} $. Note that any language in $\mathsf{L}$ is trivially recognized by a 2pfa($k$) that uses no randomness by Fact \ref{fact:2dfak-ll}, so we conclude that the class of languages recognized with unbounded error by probabilistic multihead finite automata that are restricted to use an amount of random bits that is logarithmically bounded in terms of the input length is identical to the class corresponding to the deterministic versions of these machines.

Recall that $\mathsf{RL}$ is the class of languages recognized with one-sided bounded error by logspace PTM's in polynomial time. Theorem \ref{theorem:polyrand} implies that the logarithmic-randomness and polynomial-randomness classes for single-head 2pfa's coincide. An analogous result for multiple-head 2pfa's would establish that $\mathsf{L}=\mathsf{RL}$.

Let us now turn to multihead finite-state verifiers. By the relationship with logarithmic-space Turing machines mentioned above, Equations \ref{equation:detlogspace} through \ref{equation:amp}, as well as our Theorem \ref{theorem:main}, can also be viewed as statements about the power of IPS's whose verifiers are multihead finite automata with two-way heads, so we already know that we can use this verifier model to build IPS's with zero error for $\mathsf{NL}$, and that constant randomness does not increase their power. Note that the number of heads of the verifier will depend on the language under consideration in the constructions for these characterizations.
%


If we allow arbitrarily small nonzero error and polynomial randomness, but require that at most a constant number of coin tosses can be private, we can build an IPS where the finite-state verifier is a 2pfa(2) with a halting probability of 1 for every language in $\mathsf{NL}$: The first head runs the algorithm of Lemma \ref{lemma:mhfatoip}, and the second head performs a random walk whose expected completion time is a suitably large polynomial (see, for instance, \cite{DS92}). If this walk completes before the first head announces its decision, the verifier rejects. 

\subsection{One-way heads} \label{subsection:1wayheads}

A ``stand-alone" finite automaton with $k$ one-way heads using $r$ coins can be simulated with a 1dfa($k2^r$) that simulates all the $2^r$ 1dfa($k$)'s corresponding to different coin sequences in parallel, and accepts if a majority of these 1dfa($k$)'s accept. We also know the following about these machines:

\begin{fact}
For every $k>1$, the class of languages recognized with bounded error by 1pfa($k$)'s using a constant number of coins strictly contains the class of languages recognized by 1dfa($k$)'s. \cite{Fr79A,Ku91}
\end{fact}

\begin{fact}
For  a fixed $k>1$, there exists a language recognized by a 1pfa(2) using a constant number of coins, but not by any 1nfa($k$). \cite{Fr79A}
\end{fact}

\begin{thm}
\label{theorem:1dfak}
Deterministic finite automata with multiple one-way heads can verify membership in precisely the languages in $ \bigcup_{k\geq 1}\mathsf{1NFA}(k) $ with zero error in linear time.
\end{thm}
\begin{proof}
Every language in $ \bigcup_{k\geq 1}\mathsf{1NFA}(k) $ can be assumed to have a nondeterministic one-way multihead finite automaton which recognizes it, and is guaranteed to halt in linear time. 1dfa($k$) verifiers can handle precisely the  same languages as  1nfa($k$) recognizers by definition. 
\end{proof}
By a simple extension of the proof of  Theorem \ref{theorem:oneway}, we have
\begin{thm}\label{theorem:2pfakver}
Finite automata with multiple one-way heads that use at most a constant amount of random bits independent of their input can verify membership (according to both the strong and  the weak definitions, and with either one-way or two-way communication with the prover) in precisely the languages in  $ \bigcup_{k\geq 1}\mathsf{1NFA}(k) $.
\end{thm}

When nonzero bounded error is tolerated, the construction of Theorem \ref{theorem:oneway} can be modified to obtain constant-coin 1pfa(2) verifiers that halt with probability 1 for each language in $ \bigcup_{k\geq 1}\mathsf{1NFA}(k) $, by simply using the second head as a clock. 

\section{Private alternation with fixed number of universal moves} \label{section:PAFA}
%

Reif \cite{Re79} defined the \textit{private alternating Turing machine} (PATM) to model \textit{two-person games of incomplete information}, where one of the players is allowed to hide some of its moves from the other player, as opposed to games of complete information, that are well-known to be modeled by the alternating Turing machines (ATM's) of \cite{CKS81}, with which we assume the reader to be familiar. A portion of the memory of a PATM is private to the universal states, and cannot be read when the machine is in an existential state. The extreme special case where the existential player cannot see any moves of the universal player is modeled by the \textit{blind alternating Turing machine} (BATM), which allows the universal states to change only that private portion. Language recognition by PATM's is defined similarly to that by ATM's: A string is accepted if and only if there exists a winning strategy for the existential player in the corresponding game.

These models are linked to our results by the observation that a language has an IPS with perfect completeness and just a guarantee that nonmembers will be accepted with probability less than 1 if and only if it is recognized by a PATM with the same space and time bounds as the verifier of that IPS: One simply views the coin-tosses of the verifier as universal moves, and the branchings due to the prover messages as possible existential moves of the PATM. Our Theorem \ref{theorem:main} can then be translated to

\begin{thm}
For any space bound $s(n)=O(\log n)$, the class of languages recognized by $s(n)$-space PATM's (or BATM's) that are allowed to make a constant number of universal moves equals $\mathsf{NL}$.
\end{thm}

For contrast, we recall the corresponding classes when the bound on the number of universal moves is removed below. ($\mathsf{BASPACE(s(n))}$ (resp. $\mathsf{PASPACE(s(n))}$) denotes the class of languages recognized by $s(n)$-space BATM's (resp. PATM's).)
\begin{fact}
$\mathsf{BASPACE(1)=NSPACE(n)}$. \cite{PR79}
\end{fact}

\begin{fact}
$\mathsf{PASPACE(1)=E}$. \cite{PR79}
\end{fact}

\begin{fact}
$\mathsf{BASPACE(log)=PSPACE}$. \cite{Re79}
\end{fact}

\begin{fact}
$\mathsf{PASPACE(log)=EXPTIME}$. \cite{Re79}
\end{fact}
\section{Open questions} \label{section:Conclusion}

We have been able to represent the relationship between $\mathsf{NL}$ and $\mathsf{NP}$ in the form
\[\mathsf{NL} = \mathsf{IP}\mathrm{(cons\mbox{-}space, cons\mbox{-}random\mbox{-}bits)} \subseteq \mathsf{IP}\mathrm{(log\mbox{-}space, log\mbox{-}random\mbox{-}bits)}  = \mathsf{NP}.\]
Further examination of   other classes like $\mathsf{IP}\mathrm{(cons\mbox{-}space, log\mbox{-}random\mbox{-}bits)}$ would be interesting.

Do our results for constant-space verifiers stand when a polynomial bound is imposed on the overall runtime? Every language that can be verified by a constant-ran\-domness 2pfa that halts with probability 1 is recognized in linear time by a 2nfa($k$) for some $k$. Does there exist a language in  $\mathsf{NL}$ which cannot be recognized in linear time by any  2nfa($k$)?


Tables \ref{table:recognizers} and \ref{table:verifiers} summarize some of our findings on bounded-randomness 2pfa($k$) variants, both as recognizers, and as verifiers in one-way IPS's. $\mathsf{BPTISP(poly,log)}$ denotes the class of languages recognized with bounded error by PTM's operating in polynomial time and logarithmic space. The cells marked $?_1$ and $?_2$ correspond to classes that contain the classes corresponding to the cells to their left, and are contained in the classes corresponding to the cells above them. Can one find better characterizations for these classes?

\begin{table}[h!]
	\centering
	\caption{Complexity classes associated with different settings of 2pfa variants as bounded-error recognizers.}
	\begin{tabular}{|l|c|c|c|c|}
		 \hline  \textit{randomness complexity:}  
		& 0 & $ cons$ & $ log $ & $poly$
		\\ \hline \hline
		\textbf{single-head two-way} ~~  & $\mathsf{REG}$& $\mathsf{REG}$& $\mathsf{REG}$& $\mathsf{REG}$
		\\ \hline
		\textbf{single-head one-way} & $\mathsf{REG}$& $\mathsf{REG}$& $\mathsf{REG}$& $\mathsf{REG}$ 
		\\ \hline
		\textbf{multihead two-way} & $\mathsf{L}$& $\mathsf{L}$& $\mathsf{L}$			& $\mathsf{BPTISP(poly,log)}$
		\\ \hline	
		\textbf{multihead one-way} &$\bigcup_{k\geq 1}1DFA(k)$ & $\bigcup_{k\geq 1}1DFA(k)$ & $	?_1	$	& $  ?_2 $ 
		\\[2pt] \hline		
	\end{tabular}
	\label{table:recognizers}
\end{table}

\begin{table}
	\centering
	\caption{Complexity classes associated with different settings of 2pfa variants as bounded-error verifiers.}
	\begin{tabular}{|l|c|c|}
		 \hline  \textit{randomness complexity:}  
		& 0 & $ cons$ 
		\\ \hline \hline
		\textbf{single-head two-way} ~~  & $\mathsf{REG}$& $\mathsf{NL}$ 
		\\ \hline
		\textbf{single-head one-way} & $\mathsf{REG}$& 
		$ \bigcup_{k \geq 1} \mathsf{1NFA}(k) $
		\\[2pt]
		\hline
		\textbf{multihead two-way} & $\mathsf{NL}$& $\mathsf{NL}$ 
		\\ \hline	
		\textbf{multihead one-way} &$\bigcup_{k\geq 1}\mathsf{1NFA}(k)$ & $\bigcup_{k\geq 1}\mathsf{1NFA}(k)$ 
		\\[2pt] \hline		
	\end{tabular}
	\label{table:verifiers}
\end{table}
Although we have proved that 
\[
	\mathsf{IP}_w\mathrm{(cons\mbox{-}space, cons\mbox{-}random\mbox{-}bits)} = \mathsf{IP}\mathrm{(cons\mbox{-}space, cons\mbox{-}random\mbox{-}bits)},
	\]
we are able to reduce the error probabilities to arbitrary desired positive values only for verification according to the weak definition. Is this also possible for the strong definition? 
Similarly, is there a way to reduce the error (which gets worse as the number of heads in the simulated automaton increases) of single-head verifiers with one-way access to their inputs to arbitrary desired positive values?

In their study \cite{NY09} of interactive proof systems whose verifiers are quantum finite automata (qfa's), Nishimura and Yamakami used a weak model of real-time qfa's \cite{KW97} whose stand-alone versions cannot even recognize all regular languages. They showed that letting such verifiers communicate with a prover results in a proof system which can handle all and only the regular languages. Since general qfa models \cite{Hi10,YS11A} that make full use of the nonclassical features of quantum mechanics are able to simulate any corresponding classical system easily, we conclude that one-way interactive proof systems that would use qfa's defined according to this modern approach would be able to handle all of $\mathsf{oneway\mbox{-}IP}\mathrm{(cons\mbox{-}space, cons\mbox{-}random\mbox{-}bits,rt\mbox{-}input)}$, outperforming the systems of \cite{NY09}, despite the fact that the latter allow for two-way interaction between the verifier and the prover. The study of qfa verifiers is an interesting avenue for further research.

\section*{Acknowledgements} \label{section:Acknowledgements}

We are grateful to Martin Kutrib, who helped us immensely with our questions about nfa($k$)'s. We also thank R\={u}si\c{n}\v{s} Freivalds, Taylan Cemgil, Richard Lipton, and G\"{o}kalp Demirci for their helpful answers, Alexander Rivosh for his valuable assistance with the references in Russian, and the anonymous referees for their constructive remarks.

\bibliographystyle{alpha}
\bibliography{YakaryilmazSay}

\begin{thebibliography}{HKM11}

\bibitem[CKS81]{CKS81}
Ashok~K. Chandra, Dexter~C. Kozen, and Larry~J. Stockmeyer.
\newblock Alternation.
\newblock {\em Journal of the ACM}, 28(1):114--133, 1981.

\bibitem[CL89]{CL89}
Anne Condon and Richard~J. Lipton.
\newblock On the complexity of space bounded interactive proofs.
\newblock In {\em Proceedings of the 30th Annual Symposium on Foundations of
  Computer Science}, pages 462--467, 1989.

\bibitem[CL95]{CL95}
Anne Condon and Richard Ladner.
\newblock Interactive proof systems with polynomially bounded strategies.
\newblock {\em Journal of Computer and System Sciences}, 50(3):506--518, June
  1995.

\bibitem[Con89]{Co89}
Anne Condon.
\newblock {\em Computational Models of Games}.
\newblock MIT Press, 1989.

\bibitem[Con91]{Co91}
Anne Condon.
\newblock Space-bounded probabilistic game automata.
\newblock {\em Journal of the ACM}, 38(2):472--494, April 1991.

\bibitem[Con93a]{Co93B}
Anne Condon.
\newblock The complexity of the max word problem and the power of one-way
  interactive proof systems.
\newblock {\em Computational Complexity}, 3(3):292--305, 1993.

\bibitem[Con93b]{Co93A}
Anne Condon.
\newblock {\em Complexity Theory: Current Research}, chapter The complexity of
  space bounded interactive proof systems, pages 147--190.
\newblock Cambridge University Press, 1993.

\bibitem[DS90]{DS90}
Cynthia Dwork and Larry Stockmeyer.
\newblock A time complexity gap for two-way probabilistic finite-state
  automata.
\newblock {\em SIAM Journal on Computing}, 19(6):1011--1123, 1990.

\bibitem[DS92]{DS92}
Cynthia Dwork and Larry Stockmeyer.
\newblock Finite state verifiers $\mbox{I}$: The power of interaction.
\newblock {\em Journal of the ACM}, 39(4):800--828, 1992.

\bibitem[Fre79]{Fr79A}
R\={u}si\c{n}\v{s} Freivalds.
\newblock Language recognition using finite probabilistic multitape and
  multihead automata.
\newblock {\em Problemy Peredachi Informatsii}, 15(3):99--106, 1979.
\newblock Russian.

\bibitem[Fre81]{Fr81}
R\={u}si\c{n}\v{s} Freivalds.
\newblock Probabilistic two-way machines.
\newblock In {\em Proceedings of the International Symposium on Mathematical
  Foundations of Computer Science}, pages 33--45, 1981.

\bibitem[GS86]{GS89}
Shafi Goldwasser and Michael Sipser.
\newblock Private coins versus public coins in interactive proof systems.
\newblock In {\em Proceedings of the 18th Annual ACM Symposium on Theory of
  Computing (STOC'86)}, pages 59--68, 1986.

\bibitem[Har72]{Ha72}
Juris Hartmanis.
\newblock On non-determinancy in simple computing devices.
\newblock {\em Acta Informatica}, 1:336--344, 1972.

\bibitem[Hir10]{Hi10}
Mika Hirvensalo.
\newblock Quantum automata with open time evolution.
\newblock {\em International Journal of Natural Computing Research},
  1(1):70--85, 2010.

\bibitem[HKM11]{HKM11}
Markus Holzer, Martin Kutrib, and Andreas Malcher.
\newblock Complexity of multi-head finite automata: Origins and directions.
\newblock {\em Theoretical Computer Science}, 412:83--96, 2011.

\bibitem[Ka{\c{n}}89]{Ka89}
J\={a}nis Ka{\c{n}}eps.
\newblock Stochasticity of the languages acceptable by two-way finite
  probabilistic automata.
\newblock {\em Diskretnaya Matematika}, 1:63--67, 1989.
\newblock (Russian).

\bibitem[Kar67]{Ka67}
Richard~M. Karp.
\newblock Some bounds on the storage requirements of sequential machines and
  $\mbox{Turing}$ machines.
\newblock {\em Journal of the Association for Computing Machinery},
  14(3):478--489, 1967.

\bibitem[KF90]{KF90}
J\={a}nis Ka{\c{n}}eps and R\={u}si\c{n}\v{s} Freivalds.
\newblock Minimal nontrivial space complexity of probabilistic one-way
  $\mbox{T}$uring machines.
\newblock In {\em Proceedings on Mathematical Foundations of Computer Science},
  volume 452 of {\em Lecture Notes in Computer Science}, pages 355--361, New
  York, NY, USA, 1990. Springer-Verlag New York, Inc.

\bibitem[KF91]{KF91}
J\={a}nis Ka{\c{n}}eps and R\={u}si\c{n}\v{s} Freivalds.
\newblock Running time to recognize nonregular languages by 2-way probabilistic
  automata.
\newblock In {\em Automata, Languages and Programming}, volume 510 of {\em
  Lecture Notes in Computer Science}, pages 174--185. Springer, 1991.

\bibitem[Kut91]{Ku91}
Miros{\l}aw Kuty{\l}owski.
\newblock Multihead one-way finite automata.
\newblock {\em Theoretical Computer Science}, 85(1):135--153, 1991.

\bibitem[KW97]{KW97}
Attila Kondacs and John Watrous.
\newblock On the power of quantum finite state automata.
\newblock In {\em FOCS'97: Proceedings of the 38th Annual Symposium on
  Foundations of Computer Science}, pages 66--75, 1997.

\bibitem[Mac97]{Ma97}
Ioan~I. Macarie.
\newblock Multihead two-way probabilistic finite automata.
\newblock {\em Theory of Computing Systems}, 30(1):91--109, 1997.

\bibitem[Mac98]{Ma98}
Ioan~I. Macarie.
\newblock Space-efficient deterministic simulation of probabilistic automata.
\newblock {\em SIAM Journal on Computing}, 27(2):448--465, April 1998.

\bibitem[NH71]{NH71}
Masakazu Nasu and Namio Honda.
\newblock A context-free language which is not acceptable by a probabilistic
  automaton.
\newblock {\em Information and Control}, 18(3):233--236, 1971.

\bibitem[NY09]{NY09}
Harumichi Nishimura and Tomoyuki Yamakami.
\newblock An application of quantum finite automata to interactive proof
  systems.
\newblock {\em Journal of Computer and System Sciences}, 75:255--269, 2009.

\bibitem[PR79]{PR79}
Gary~L. Peterson and John~H. Reif.
\newblock Multiple-person alternation.
\newblock In {\em Proceedings of the 20th Annual Symposium on Foundations of
  Computer Science (FOCS'79)}, pages 348--363. IEEE Computer Society, 1979.

\bibitem[Rei79]{Re79}
John~H. Reif.
\newblock Universal games of incomplete information.
\newblock In {\em Proceedings of the Eleventh Annual ACM Symposium on Theory of
  Computing (STOC'79)}, pages 288--308. ACM, 1979.

\bibitem[SB96]{SB96}
Jeffrey Shallit and Yuri Breitbart.
\newblock Automaticity $\mbox{ I}$: Properties of a measure of descriptional
  complexity.
\newblock {\em Journal of Computer and System Sciences}, 53:10--25, 1996.

\bibitem[Sha92]{Sh92}
Adi Shamir.
\newblock \mbox{IP = PSPACE}.
\newblock {\em Journal of the ACM}, 39(4):869--877, October 1992.

\bibitem[SY12]{SY12}
A.~C.~Cem Say and Abuzer Yakaryilmaz.
\newblock Finite state verifiers with constant randomness.
\newblock In {\em CiE}, volume 7318 of {\em LNCS}, pages 646--654. Springer,
  2012.

\bibitem[YS11]{YS11A}
Abuzer Yakary{\i}lmaz and A.~C.~Cem Say.
\newblock Unbounded-error quantum computation with small space bounds.
\newblock {\em Information and Computation}, 209(6):873--892, 2011.

\end{thebibliography}

\appendix

\section{The proof of Theorem \ref{theorem:polyrand}} \label{appendix:one}

Our proof of Theorem \ref{theorem:polyrand} is based on the following 
\cite{DS90,KF91}

\begin{fact}\label{fact:ds90}
For any polynomial $p$, 2pfa's with expected runtime $O(p(n))$ recognize  only the regular languages with bounded error.
\end{fact}

We start by noting that no such program which respects a worst-case bound $b(n)$ on the number of random bits that it uses can possibly have a computational path in which a configuration of the form $(r,i)$, where $r$ is a coin-tossing state,  and $i$ is a head position, repeats. Therefore, $O(n)$ is a tight
bound on the number of usable random bits under a worst-case regime.  It is also clear that contiguous subsequences of configurations with deterministic states can have at most  linear length in halting computational paths. Therefore, all halting paths of such a machine have worst-case runtime $O(n^{2})$. Furthermore, any nonhalting path must  enter an infinite loop of deterministic configurations in  $O(n^{2})$ steps.

When $b(n)$ is a bound on the \textit{expected} number of coin tosses, it has a similar relationship with the runtime. Any halting path that tosses $k$ coins has length $O(kn)$. Any nonhalting path with nonzero probability must toss only a finite number ($k$) of coins, so it must enter an infinite deterministic loop within $O(kn)$ steps. So the expected runtime of the halting paths is $O(n b(n))$.

We could use Fact \ref{fact:ds90} directly to prove Theorem \ref{theorem:polyrand} if we had a guarantee that the machines we consider have polynomial expected time. There is no such comfort, however, since it is easy to demonstrate cases where a sizable ratio of computational paths do not halt, and expected time is therefore not bounded.

For this reason, we look at the proof of Fact \ref{fact:ds90} in detail. One starts by defining a quantitative measure of the nonregularity of a language $L\subseteq \Sigma^{*}$. For a positive integer $n$, two strings $w, w' \in \Sigma^{*}$ are  \textit{n-dissimilar}, written $w \nsim_{L,n} w'$,  if $|w| \leq n$, $|w'| \leq n$, and there exists a \textit{distinguishing string} $v \in \Sigma^{*}$ with $|wv| \leq n$, $|w'v| \leq n$, and $wv \in L$ iff $w'v \notin L$. Let $N_{L}(n)$ be the maximum $k$ such that there exist $k$ distinct strings that are pairwise $\nsim_{L,n}$. It can be shown  \cite{Ka67,KF90,SB96} that 

\begin{fact}\label{fact:ndiv2}
If $L$ is not regular, then $N_{L}(n)\geq \frac{n}{2} + 1$ for infinitely many $n$.
\end{fact}

In the rest of the proof, Dwork and Stockmeyer \cite{DS90} develop a technique for constructing a Markov chain $P_{A,xy}$ with $2c$ states that models the computation of a given 2pfa $A$ with $c$ states on the concatenated string $xy$, where $x$ and $y$ are given strings. State
1 of the Markov chain corresponds to $M$ being at the beginning of its computation on the last symbol of $\cent x$. (Note that every 2pfa can be modified to start here, without changing the recognized language. The nonhalting states of the modified 2pfa are $\{q_{1},q_{2},\ldots,q_{c-1}\}$.) For $ 1\leq j \leq c-1 $, state $ j $ of the Markov chain corresponds to $M$ being in the configuration with the machine in state $ q_{j} $ and the head on the last symbol of $\cent x$, and state $ c+j-1 $ corresponds to $M$ being in the configuration with the machine in state $ q_{j} $ and the head on the first symbol of $ y \cent$. State $2c-1$ corresponds  to a disjunction of rejection, infinite loop with the head never leaving the region $\cent x$, and infinite loop within the region $y \cent$. State $2c$ corresponds to acceptance. The probability that $P_{A,xy}$ is absorbed in state $2c$ when started in state 1 equals the probability that $A$ accepts $xy$. 

The proof then considers any 2pfa  $M$ that recognizes language $L$ in expected time $T(n)$, and proceeds to establish a lower bound, in terms of $N_{L}(n)$, on $T(n)$. This is accomplished  by  showing that, for sufficiently large values of $n$, if the desired lower bound does not exist, then there must be two  pairwise $\nsim_{L,n}$ strings $w$ and $w'$, with distinguishing string $v$, such that the  Markov chains $P_{M,wv}$ and $P_{M,w'v}$ are  ``too close'' according to a notion of closeness defined in \cite{DS90}. 
In a step crucial  for our purposes, (\cite{DS90}, page 1015, Lemma 4.2,) it is proven that, if $P_{M,wv}$ and $P_{M,w'v}$ are so close, and if it is guaranteed that both Markov chains are absorbed to state $2c-1$ or $2c$ with total probability 1 within expected time $T(n)$, then the acceptance probabilities of $wv$ and $w'v$ must be so close that they must  both be members (or non-members) of $L$, contradicting their $n$-dissimilarity, thereby establishing the desired bound on $T(n)$. Fact \ref{fact:ds90} is then obtained by combining this result with Fact \ref{fact:ndiv2}.

In the case of our machines, the Markov chains produced according to the construction mentioned above do turn out to be that close to each other, but they do not necessarily satisfy the guarantee of absorption to state $2c-1$ or $2c$, since $ M $ can possibly enter an infinite loop in which the head shuttles back and forth over the $\cent w$ (resp., $\cent w'$) and $v \cent$ regions.
Fortunately, this problem goes away on a careful look: Consider all cycles of transitions that have probability 1 between the  two regions in the produced Markov chains. The set of states appearing in such a cycle is an absorbing class. Furthermore, it is certain that absorption to either state $2c-1$, or $2c$, or one of these loop classes will take place within expected time $n b(n)$, and the acceptance probability would not change if one redirected these transitions to state $ 2c-1 $, so Lemma 4.2 of \cite{DS90} still applies, and we can conclude with the same reasoning as in the proof of Fact \ref{fact:ds90} that if $b(n)$ is polynomially bounded, then $L$ is regular.

\end{document}